\documentclass[a4paper, 11pt]{article}  % add "twoside" for page waving"
\usepackage{algorithm}
\usepackage[noend]{algpseudocode}
\algnewcommand{\algorithmicassumption}{\textbf{Requirement:}}
\algnewcommand{\Assume}{\item[\algorithmicassumption]}
\algrenewcomment[1]{\hfill \(\triangleright\) \emph{\small #1}} % to italicize
\algnewcommand{\InlineIf}[2]{% single line if-then
  \algorithmicif\ #1\ \algorithmicthen\ #2}
\algnewcommand{\InlineElse}[1]{% single line else
  \algorithmicelse\ #1}
\algnewcommand{\InlineIfElse}[3]{% single line if-then-else
  \algorithmicif\ #1\ \algorithmicthen\ #2\ \algorithmicelse\ #3}
\algnewcommand{\InlineFor}[2]{\algorithmicfor\ #1\ \algorithmicdo\ #2}
\algnewcommand{\CommentLine}[1]{\(\triangleright\) \emph{\small #1}}

\algnewcommand{\algorithmicand}{\textbf{and}}
\algnewcommand{\algorithmicor}{\textbf{or}}
\algnewcommand{\FOR}{\algorithmicfor}
\algnewcommand{\OR}{\algorithmicor}
\algnewcommand{\AND}{\algorithmicand}
\algnewcommand{\IF}{\algorithmicif}
\algnewcommand{\THEN}{\algorithmicthen}
\algnewcommand{\ELSE}{\algorithmicelse}
\algnewcommand{\True}{\textsc{True}}
\makeatletter
\newcommand{\StateX}[1]{%
  \setlength\@tempdima{\algorithmicindent}%
  \Statex\hskip\dimexpr#1\@tempdima\relax}
\makeatother
\makeatletter
\newcommand{\algoCaptionLabel}[3]{
     \caption[\textproc{#1}]{\textproc{#1}\ifthenelse{\equal{#2}{}}{}{$(#2)$} }%
     \NR@gettitle{\textproc{#1}}%
      \label{algo:#3}
     }%
\makeatother

\usepackage[utf8]{inputenc}
\usepackage{pgfplots}
\pgfplotsset{compat=newest}
\usepgfplotslibrary{fillbetween}
\usetikzlibrary{patterns}
\usetikzlibrary{arrows.meta}    % required for custom arrowheads in tikz
\usepackage{amsmath}
\usepackage{amssymb}
\usepackage{mathtools}	% for \coloneqq and \prescript
\usepackage{hyperref}
\usepackage{a4wide}
\usepackage{array}
\usepackage{titling}
\usepackage{setspace}
\usepackage{geometry}   % [left=2cm,right=2cm,top=2cm,bottom=3cm]
\hypersetup{						% for colours of clickable stuff
    colorlinks,				% idk but necessary
    citecolor=cyan,			% colour for \cite{} (NOT brackets, number only)
    filecolor=red,			% colour for internal file reference
    linkcolor=green!60!black,	% colour for \ref{} and \nameref{} and stuff
    urlcolor=black			% colour for \url{}
}
\usepackage{comment}
\usepackage{amsthm}
\usepackage[capitalize,nameinlink]{cleveref}	% for smart clickable stuff 
\usepackage{bm}			% for bold math symbols
\usepackage{mathrsfs}
\usepackage{authblk}
\usepackage{multirow}
\usepackage{sectsty}

\newcommand{\new}[1]{\textcolor{black}{#1}}
\newcommand{\mathnew}[1]{\mathcolor{black}{#1}}
\newcommand{\secnew}[1]{
    \sectionfont{\color{black}}
    \section{#1}
    \sectionfont{\color{black}}
}

\newtheorem{theorem}{Theorem}[]
\newtheorem{lemma}[theorem]{Lemma}
\AddToHook{env/lemma/begin}{\crefalias{theorem}{lemma}}
\Crefname{lemma}{Lemma}{Lemmas}

\AddToHook{env/corollary/begin}{\crefalias{theorem}{corollary}}
\Crefname{corollary}{Corollary}{Corollaries}

\AddToHook{env/definition/begin}{\crefalias{theorem}{definition}}
\Crefname{definition}{Definition}{Definitions}

\AddToHook{env/proposition/begin}{\crefalias{theorem}{proposition}}
\Crefname{proposition}{Proposition}{Propositions}
\AddToHook{env/example/begin}{\crefalias{theorem}{example}}
\Crefname{example}{Example}{Examples}
\newtheorem{remark}[theorem]{Remark}
\AddToHook{env/remark/begin}{\crefalias{theorem}{remark}}
\Crefname{remark}{Remark}{Remarks}

\makeatletter
\AddToHook{env/algorithmic/begin}{\def\@currentcounter{ALG@line}}
\makeatother

\crefalias{ALG@line}{line}

\allowdisplaybreaks
\title{Refined bit complexity for the computation of at least one point per
connected component of a smooth complete intersection real algebraic set}
\author[1]{Jesse Elliott} 
\author[1]{Mark Giesbrecht} 
\author[2]{Edern Gillot} 
\author[2]{Mohab {Safey El Din}}
\author[1]{Éric Schost}
\affil[1]{David R. Cheriton School of Computer Science, University of Waterloo,
ON, Canada}
\affil[2]{Sorbonne Université, CNRS, LIP6, Paris, France}

\date{April 10, 2026}

\DeclareMathOperator{\htt}{ht}
\DeclareMathOperator{\polar}{W}
\DeclareMathOperator{\jac}{jac}
\DeclareMathOperator{\GL}{GL}

\newcommand{\Nintegers}{\mathbb{N}}
\newcommand{\Zintegers}{\mathbb{Z}}
\newcommand{\ratfield}{\mathbb{Q}}
\newcommand{\realfield}{\mathbb{R}}
\newcommand{\compfield}{\mathbb{C}}

\newcommand{\bff}{\bm{f}}
\newcommand{\bd}{\bm{d}}
\newcommand{\bn}{\bm{n}}
\newcommand{\bx}{\bm{x}}
\newcommand{\by}{\bm{y}}
\newcommand{\bu}{\bm{u}}

\newcommand{\bmeta}{\bm{\eta}}
\newcommand{\bsigma}{\bm{\sigma}}

\newcommand{\bA}{{\bm{A}}}
\newcommand{\scrD}{\mathscr{D}}
\def\polarsys#1#2{\mathcal{S}_{#1}(#2)}

\begin{document}
%\pagenumbering{gobble}

\maketitle

\begin{abstract}
    We refine the bit complexity analysis of an algorithm for the computation of
    at least one point per connected component of a smooth real algebraic set,
    yielding exponential speedup (with respect to the number of variables)
    compared to prior works. The algorithm which is analyzed is based on the
    critical point method, reducing the problem to computations of critical
    points associated to the restriction of generic projections on lines to the
    studied variety.

    Our refinement, and the subsequent improved complexity statement, comes from
    a better utilization of the multi-affine structure of polynomial systems
    encoding these sets of critical points. The bit-size estimates on the size
    of the output produced by this algorithm are also improved by this
    refinement. 
\end{abstract}

\section{Introduction}

Computing at least one point in each connected component of a real algebraic or
a semi-algebraic set is a classical computational problem of real algebraic
geometry, which is used as a subroutine of higher-level algorithms such as
quantifier elimination (see e.g.~\cite{basupollackroy, Le_2021}), real root
classification (see e.g.~\cite{LeSa22, GaSa24}) or polynomial optimization (see
e.g.~\cite{GreSa14}).

As in~\cite{elliottgiesbrechtschostgeneral}, we focus on the
problem of computing at least one point per connected component of a real
algebraic set defined as the vanishing set of a sequence of polynomials 
\(\bff = (f_1, \ldots, f_{p})\subset \ratfield[x_1, \ldots,
x_{n}]\) satisfying the following regularity assumptions: 
\begin{itemize}
    \item the ideal generated by \(\langle f_1, \ldots, f_p \rangle\) has
    codimension \(p\), 
    \item at any point \(\by\) of the complex algebraic set  \(V = V\left( f_1, 
    \ldots, f_p \right)\subset \compfield^{n} \), the Jacobian matrix 
    \[
    \begin{bmatrix} 
        \frac{\partial f_1}{\partial x_1} & \cdots & \frac{\partial
            f_1}{\partial x_{n}} \\
        \vdots & \ddots & \vdots \\
        \frac{\partial f_p}{\partial x_1} & \cdots & \frac{\partial
            f_p}{\partial x_{n}} 
    \end{bmatrix} 
    \]
    has rank \(p\). Equivalently, \(V\) is smooth and equidimensional, \new{and
    \(\langle f_1, \ldots, f_p \rangle\) is radical}.  
\end{itemize}
These properties are satisfied generically. 
We call real algebraic sets defined by a polynomial sequence satisfying such a
genericity assumption smooth complete intersection real algebraic sets. 

The goal of this note is to refine the complexity analysis done in
\cite{elliottgiesbrechtschostgeneral} on an algorithm which is a
variant of \cite{safeyschostonepoint} for computing at least one point
in each connected component of a smooth complete intersection real
algebraic set, which is the real trace \(V\cap\realfield^{n}\) of an
algebraic set \(V\).  \new{These algorithms reduce the problem to the
  one of solving systems with finitely many complex solutions.}
  \new{This refinement is based on nowadays classical degree and
    height bounds which lie at the interplay of computer algebra and
    intersection theory, by leveraging some special properties of
  those systems with finitely many complex solutions.}

The algorithm \cite[Algorithm 1]{elliottgiesbrechtschostgeneral} starts by
performing a randomly chosen linear change of coordinates, which is expectedly
generic enough to let some geometric subsets of \(V\) satisfy Noether position
properties. In \cite{safeyschostonepoint, elliottgiesbrechtschostgeneral}, these
changes of variables are indicated with a matrix as an upper script. We omit
this notation in this note. 

These subsets of \(V\) are sets of critical points (also known as polar
varieties) of the restriction to \(V\cap\realfield^{n}\) of the canonical
projections \(\pi_i : \left( x_1, \ldots, x_{n} \right) \mapsto \left( x_1,
\ldots, x_{i} \right) \) for \(1\le i \le n-p =\new{\dim(V)}\). As
in~\cite{elliottgiesbrechtschostgeneral}, we denote them by \(\polar(i, \bff)\).
By convention, \(\polar(\new{\dim(V)}+1, \bff) = V\). Then, the algorithm from
\cite{safeyschostonepoint} computes a symbolic representation of the sets
\[
    \polar(i, \bff)\cap \pi_{i-1}^{-1}(\mathbf{0}) \quad \text{ for } \quad 
    2 \le i \le \new{\dim(V)}+1
\]
with \(\polar(1, \bff)\), all of them being finite thanks to the linear change
of coordinates. The fibers of the \(\pi_i\)'s are taken at the origin but can be
taken at any arbitrary point of their target spaces
(see \cite[Theorem 2]{safeyschostonepoint}). 

In \cite{safeyschostonepoint}, the sets \(\polar(i, \bff)\) are defined by the
vanishing of entries of \(\bff\) and of the maximal minors of the truncated 
Jacobian matrix
\[
\jac(\bff, i) = 
\begin{bmatrix} 
    \frac{\partial f_1}{\partial x_{i+1}} & \cdots & \frac{\partial
    f_1}{\partial x_{n}} \\
            \vdots & \ddots & \vdots \\
            \frac{\partial f_p}{\partial x_{i+1}} & \cdots & \frac{\partial
    f_p}{\partial x_{n}} 
\end{bmatrix} 
.\] 

Note that \(\polar(i, \bff)\) is 
equivalently defined as the projection on
the \(\left( x_1, \ldots, x_{n} \right) \)-space of the solutions to the 
Lagrange systems.
\[
    f_1 = \cdots = f_p = 0, \quad [\ell_1, \ldots, \ell_p]\jac(\bff, i) =
    \mathbf{0}, \quad u_1\ell_1 + \cdots + u_{p}\ell_{p} = 1,
\]
where \(\ell_1, \ldots, \ell_p\) are new indeterminates (the Lagrange
multipliers) and the \((u_1, \ldots, u_p)\in \ratfield\) are randomly chosen (to
satisfy a generic property).  

This is used in e.g.~\cite{safeyschostcomplexity}
for solving optimization problems, or in \cite{SaSc17} for answering
connectivity queries in real algebraic sets, as well as in  
\cite{elliottgiesbrechtschostgeneral} for computing sample points in
real algebraic sets.

Moreover, in~\cite{elliottgiesbrechtschostgeneral}, the fibers of the
projections \(\pi_i\) are taken over \emph{randomly} chosen points
\(\bsigma_i\in \ratfield^{i}\) in order to ensure regularity properties of their
intersection with \(\polar(i, \bff)\) (basically, the obtained polynomial system
generates a radical ideal). This yields~\cite[Algorithm
1]{elliottgiesbrechtschostgeneral}.

A few more comments on the structure of the above Lagrange systems are in order
here, since we use it to improve the complexity analysis
of~\cite{elliottgiesbrechtschostgeneral}. To each equation defining \(\polar(i,
\bff)\cap \pi_{i-1}^{-1}(\bsigma_{i-1})\) (up to elimination of the Lagrange
multipliers), one associates a \emph{bi-degree} \((d_1, d_2)\in \Nintegers\times
\Nintegers\) which is the couple of partial degrees in respectively the block of
variables \(x_1, \ldots, x_{n}\) and \(\ell_1, \ldots, \ell_p\). We denote by
\(d\) the maximum degree of the total degrees of the input polynomials in
\(\bff\). 
Hence, \(\polar(i, \bff)\cap \pi_{i-1}^{-1}(\bsigma_{i-1})\) is defined by 
\begin{itemize}
    \item \(p\) equations of bi-degree  \((d, 0)\) (the vanishing of the
        \(f_i\)'s),
    \item  \(n-i\) equations of bi-degree  \((d-1, 1)\) (coming from multiplying
        the row vector of Lagrange multipliers with \(\jac(\bff, i)\)), 
    \item  \(1\) equation of bi-degree  \((0, 1)\) (the affine linear form in
        \( \ell_1, \ldots, \ell_p\)), 
    \item \(i-1\) equations of bi-degree \((1, 0)\) (coming from
        \(\pi_{i-1}^{-1}(\bsigma_{i-1})\)). 
\end{itemize}
Further, we denote by \(\bd_i\) the sequence of the bi-degrees of these
polynomials and by  \(\bn\) the couple  \(n, p\). 

These modifications allow us to use the symbolic homotopy algorithm in
\cite[Algorithm 2]{safeyschostcomplexity} for computing ``regular'' solutions to 
zero-dimensional polynomial systems leveraging multi-affine structures
which are known to the end-user. For systems generating a radical ideal, all
solutions are computed and symbolically represented by a zero-dimensional 
\emph{rational parametrization}. We recall the definition of such a
data structure below.    

\medskip
%\begin{definition}\label{def:ratparam0D}
    Let $S \subset \mathbb{C}^n$ be a finite algebraic set defined by
    polynomials with rational coefficients. A \emph{{zero-dimensional
    rational parametrization}} $((w, v_1, \dots, v_n), \lambda)$ of $S$ consists
    of: 
    \begin{itemize} 
        \item polynomials $w, v_1, \dots, v_n \in \mathbb{Q}[T]$
        for a new variable $T$, such that $w$ is monic and squarefree, and
        $\deg(v_i) < \deg(w)$ for all $i$, 
    \item a $\mathbb{Q}$-linear form
        $\lambda$ in $n$ variables such that $\lambda(v_1, \dots, v_n) = Tw'
        \mod w$, 
    \end{itemize} 
    such that the equality $S =
        \left\{\left(\frac{v_1(t)}{w'(t)}, \cdots, \frac{v_n(t)}{w'(t)}\right)
        \middle| \ w(t) = 0\right\}$ is satisfied.
%\end{definition}

\medskip
Such a data structure is a standard way to represent outputs of algorithms that
lie in $\overline{\mathbb{Q}}^n$, such as the algorithm we analyze. See
\cite{elliottgiesbrechtschostgeneral} and references therein for more details.

\medskip
The tour de force, performed in \cite{elliottgiesbrechtschostgeneral}, is to
carefully analyze and quantify probabilistic aspects of the algorithm outlined
above, i.e. control the degrees of the Zariski closed subsets where one should
\emph{not} pick the randomly chosen linear change of variables, the points
\(\bsigma_i\) used to define the fibers of the \(\pi_i\)'s, and the vectors
\(u_1, \ldots, u_p\) used to define the Lagrange systems. This allows us to
control the probability of success of random choices for these objects as well
as their bit sizes. These estimates are combined with the bit complexity
statement \cite[Theorem 1]{safeyschostcomplexity} to provide a bit complexity
analysis of \cite[Algorithm 1]{elliottgiesbrechtschostgeneral} with degree and
\emph{height} bounds on the output.  

In this note, we show how to leverage even more the bi-affine structure of the
Lagrange systems defining \(\polar(i, \bff)\) to obtain a better bit complexity
statement for the same algorithm \cite[Algorithm
1]{elliottgiesbrechtschostgeneral}, as well as degree and height bounds. Before
stating our main result, we recall the definition of height. 

\medskip
We use $\log(x)$ to mean the logarithm of \(x\) in base $2$.
%\begin{definition}\label{def:height}
    The \emph{{height}} of a non-zero $\frac{\new{a}}{\new{c}} \in \mathbb{Q}$,
    for \(\new{c}>0\) prime with \(\new{a}\),
    denoted $\htt\left(\frac{\new{a}}{\new{c}}\right)$, is the maximum of
    $\log|\new{a}|$ and $\log(\new{c})$. If $f$ is a non-zero polynomial with
    rational coefficients, define $v \in \mathbb{N}$ as the least common
    multiple of the denominators of all coefficients of $f$. The \emph{{height}}
    of $f$, denoted $\htt(f)$, is the maximum of $\log(v)$ and the heights of
    the coefficients of $vf$.
%\end{definition}

\medskip
We can now state our main results. Below, we use the \(\widetilde{O}\) notation
to omit poly-logarithmic factors. \new{For clarity, we split the result in two
statements, dealing respectively with the cases where $d>2$ and $d=2$.}

\begin{theorem}\label{thm:main}
Let $\bff \coloneqq (f_1, \dots, f_p)\subset \Zintegers[x_1, \ldots, x_{n}]$
with $\deg(f_i) \leq d$ and $\htt(f_i) \leq b$ for all $1 \le i\le p$. Suppose
that the ideal generated by $\bff$ is radical and that $V \coloneqq {V}(\bff)
\subset \mathbb{C}^n$ is smooth and $(n-p)$-equidimensional. Let $0 < \epsilon <
1$ be a real number, and suppose that $d > 2$.

There exists a randomized algorithm taking $\bff$ and $\epsilon$ as inputs and
outputs $n-p+1$ zero-dimensional rational parametrizations, whose union of
zeroes contains at least one point per connected component of $V \cap
\mathbb{R}^n$, with probability at least $1 - \epsilon$. Otherwise, the
algorithm either returns a proper subset of those points, or \emph{FAIL}. In any
case, the algorithm performs
\[\widetilde{O}\left(\log(1/\epsilon)(b + d\log(1/\epsilon))
\binom{n}{p}^2\binom{n+d}{d} d^{2p}(d-1)^{2(n-p)}\right)\] 
bit operations. The polynomials in the output have degree at most
$d^p(d-1)^{n-p}\binom{n-1}{p-1}$, and height in
\[\widetilde{O}\left((b + d\log(1/\epsilon)) d^p
(d-1)^{n-p}\binom{n}{p}\right).\]
\end{theorem}

\new{We now state our result when $d=2$.}

\begin{theorem}\label{thm:main2}
\new{Let $\bff \coloneqq (f_1, \dots, f_p)\subset \Zintegers[x_1, \ldots,
x_{n}]$ with $\deg(f_i) \leq 2$ and $\htt(f_i) \leq b$ for all $1 \le i\le p$.
Suppose that the ideal generated by $\bff$ is radical and that $V \coloneqq
{V}(\bff) \subset \mathbb{C}^n$ is smooth and $(n-p)$-equidimensional. Let $0 <
\epsilon < 1$ be a real number.}

\new{There exists a randomized algorithm taking $\bff$ and $\epsilon$ as inputs
and outputs $n-p+1$ zero-dimensional rational parametrizations, whose union of
zeroes contains at least one point per connected component of $V \cap
\mathbb{R}^n$, with probability at least $1 - \epsilon$. Otherwise, the
algorithm either returns a proper subset of those points, or \emph{FAIL}. In any
case, the algorithm performs
\[\widetilde{O}\left(n^6(n-p+1)\log(1/\epsilon)(b + \log(1/\epsilon) + n)
\binom{n}{p}^2 2^{2p}\right)\] 
bit operations. The polynomials in the output have degree at most
$2^p\binom{n-1}{p-1}$, and height in
\[\widetilde{O}\left(n^2(b + \log(1/\epsilon) + n) 2^p \binom{n}{p}\right).\]}
\end{theorem}

\new{Compared to other algorithms solving this problem, and using the same
notation for number of polynomials, variables, degrees and heights, our
complexity estimate is as follows:}

\begin{table}[H]
\new{
\centering
\begin{tabular}{|c|c|}
    \hline
    Algorithm & Complexity \\
    \hline
    \rule{0pt}{3ex} Cylindrical Algebraic Decomposition \cite{collins} &
    $b(pd)^{2^{O(n)}}$ (general, deterministic) \\
    \rule{0pt}{4ex} \cite[Algorithm 12.17]{basupollackroy} & $bd^{O(n)}$
    (general, deterministic)\\
    \rule{0pt}{4ex} \multirow{4}{*}{Polar varieties-based algorithms 
    \cite{safeyschostonepoint,BGHM97,BGHM01,BGHP05,BGHSS10}} & 
    $\widetilde{O}\big(\binom{n+d}{d}\binom{n-1}{p}^2(pd)^{3n+3}\big)$ \\
    & (regularity assumptions, \\
    & arithmetic, probabilistic) \\
    \rule{0pt}{4ex} \multirow{3}{*}{ \cite[Algorithm
    1]{elliottgiesbrechtschostgeneral}} &
    $\widetilde{O}\big((b+\log(1/\epsilon))\log(1/\epsilon)d^{3n+2p+1}\big)$ \\
    & (regularity assumptions, \\
    & probability of success $\geq 1-\epsilon$) \\
    \rule{0pt}{5ex} \multirow{4}{*}{This work} & $\widetilde{O}\Big(n^4
    (n-p+1)(b + d\log(1/\epsilon) + dn)$ \\
    & $\log(1/\epsilon)\binom{n}{p}^2\binom{n+d}{d}d^{2p}(d-1)^{2n-2p}\Big)$ \\
    & (regularity assumptions, \\
    & probability of success $\geq 1-\epsilon$) \\
    \hline
\end{tabular}
}
\end{table}

\begin{remark}
    The estimates in \emph{\cite{elliottgiesbrechtschostgeneral}} were
    exponential in \(p\). Precisely, the exponential behavior in the number of
    bit operations was in  \(d^{3n+2p+1}\) (instead of \(d^{2p}
    (d-1)^{2(n-p)}\)), leading to a significant overestimation in regimes where
    \(p\simeq \frac{n}{2}\).   
    Similar statements can be made on the degree and the height bounds on the
    output. 

    Note also that our complexity estimates are polynomial in \(n\) when \(d=2\)
    and \(p\) is fixed, which was not the case for the estimates in
    \emph{\cite{elliottgiesbrechtschostgeneral}}.
\end{remark}

\secnew{Algorithm}

\new{
We briefly recall in this section the variant of
\cite{safeyschostonepoint} given in \cite[Algorithm
1]{elliottgiesbrechtschostgeneral}
}

\begin{algorithm}[H]
\algoCaptionLabel{\cite[Algorithm 1]
{elliottgiesbrechtschostgeneral}}{}{main}
\begin{algorithmic}[1]
    \Require $\bff = (f_1, \dots, f_p) \in \Zintegers[x_1, \dots, x_n]^p$ and $0
    < \epsilon < 1$.
    \Assume $\deg(f_i) \leq d$, $\htt(f_i) \leq b$, and $d \geq 2$. 
    \Ensure $n-p+1$ zero-dimensional rational parametrizations, whose union of
    zeroes includes at least one point in each connected component of $V(\bff)
    \cap \realfield^n$, with probability at least $1-\epsilon$. 
    \State Construct
    \vspace{-0.5cm}
    \begin{align*}
        S &\coloneqq \{1,2,\dots, \lceil20\epsilon^{-1}n^3(2d)^{5n}\rceil\}, \\
        T &\coloneqq \{1,2,\dots, \lceil4\epsilon^{-1}nd^{4n}\rceil\}, \\
        R &\coloneqq \{1,2,\dots, \lceil4\epsilon^{-1}nd^{2n}\rceil\}, \\
    \end{align*}

    \vspace{-0.7cm}
    \noindent and choose $\bA \in S^{n^2}$, $\bsigma \in T^{n-p}$ and $\bu \in
    R^p$ uniformly at random;
    \For{$i \in \{1, \dots, n-p+1\}$}

        \State \parbox[t]{370pt}{%
        Build a straight-line program $\Gamma_i$ that computes the
        polynomials
        \[x_1 - \sigma_1, \dots, x_{i-1}-\sigma_{i-1},\ \bff^{\bA}, \ 
        \begin{bmatrix} L_1 \cdots L_p \end{bmatrix} \cdot \jac(\bff^{\bA}, i), 
        \ u_1L_1 + \cdots + u_pL_p - 1;\]\strut
        \vspace{-0.5cm}}
        \State Run \cite[Algorithm 2]{safeyschostcomplexity} $k \geq
        \log_2(4n/\epsilon)$ times with input $\Gamma_i$;\label{lin}
        \State \parbox[t]{313pt}{%
        Take $\scrD_i = ((w_i, v_{i,1}, \dots, v_{i,n+p}),\lambda_i)$ to
        be the highest cardinality zero-dimensional rational parametrization
        returned in Step 4;\strut}
        \State \parbox[t]{313pt}{%
        Denote $\scrD_i' = ((w_i, v_{i,1}, \dots, v_{i,n}),\lambda_i)$
        the parametrization of the projection of $\scrD_i$ onto the $x_1, \dots,
        x_n$-space;\strut}
    \EndFor
    \State \Return $[\scrD_1', \dots, \scrD_{n-p+1}']$.
\end{algorithmic}
\end{algorithm}
\new{Geometrically, it is rather simple; as explained before, it applies some
randomly chosen linear change of coordinates to satisfy generic properties,
computes the critical points of a projection on a coordinate, and then
instantiates said coordinate and repeats, as illustrated on an example below.}

\begin{figure}[H]
\begin{picture}(\textwidth,4.5cm)
% Width: 120pt, full size: 0 --- 410 pt
%\put(50,57.5){Step 1:}
\put(30,0){
    \includegraphics{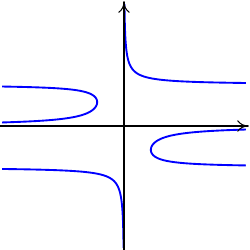}
}
\put(177.5,56.5){\tikz\draw[-{Latex[length=3mm]},line width=0.4mm](0,0) 
to (2,0);}
\put(260,0){
    \includegraphics{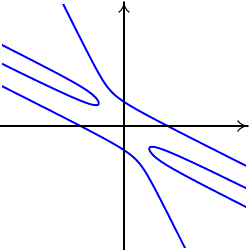}
}
\end{picture}
\caption{\new{Illustration of the change of variables step applied to $\bff = 
(4x_1(x_2^3-x_2)-1)$.}}
\label{fig:step1}
\end{figure}

\begin{figure}[H]
\begin{picture}(\textwidth,4.5cm)
% Width: 120pt, full size: 0 --- 410 pt
%\put(50,57.5){Step 1:}
\put(0,0){
    \includegraphics{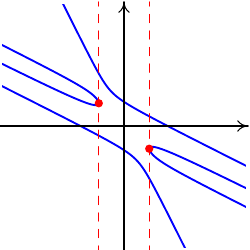}
}
\put(125,56.5){\tikz\draw[-{Latex[length=3mm]},line width=0.4mm](0,0) to
(0.75,0);}
\put(145,0){
    \includegraphics{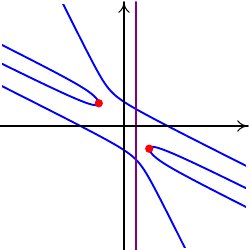}
}
\put(270,56.5){\tikz\draw[-{Latex[length=3mm]},line width=0.4mm](0,0) to
(0.75,0);}
\put(290,0){
    \includegraphics{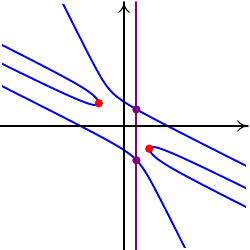}
}
\end{picture}
\caption{\new{Illustration of the computation of critical points and
instantiation step.}}
\label{fig:step2}
\end{figure}

\begin{figure}[H]
\begin{picture}(\textwidth,4.5cm)
% Width: 120pt, full size: 0 --- 410 pt
%\put(50,57.5){Step 1:}
\put(30,0){
    \includegraphics{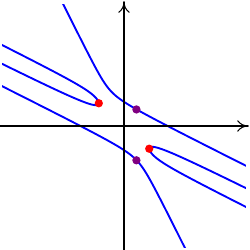}
}
\put(177.5,56.5){\tikz\draw[-{Latex[length=3mm]},line width=0.4mm](0,0) 
to (2,0);}
\put(260,0){
    \includegraphics{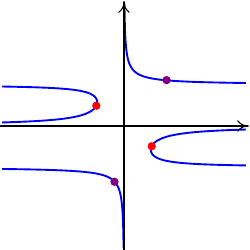}
}
\end{picture}
\caption{\new{Illustration of the change of variables reversal step.}}
\label{fig:step3}
\end{figure}

\new{The most costly step of this algorithm is \Cref{lin}, which computes the
parametrization of the critical points. All other steps are for picking randomly
chosen objects, or involve elementary manipulations of such objects.}

\section{Proof of \texorpdfstring{\cref{thm:main,thm:main2}}{}}

The bulk of the proof is to refine degree and height bounds on the set of
solutions to the aforementioned Lagrange systems, by leveraging their structure
which we emphasized above. All other ingredients (height bounds for the objects
which are randomly chosen and probability estimates) come from
\cite{elliottgiesbrechtschostgeneral}. 

\Cref{lin} calls \cite[Algorithm 2]{safeyschostcomplexity} $O\left(\log(n) +
\log(1/\epsilon)\right)$ times on the input system, for $1 \leq i \leq n - p +
1$:
\[\polarsys{i}{\bff^\bA, \bu} \coloneqq \left( x_1 - \sigma_1, \dots, x_{i-1} -
\sigma_{i-1}, \bff^{\bA}, [\ell_1 \cdots \ell_p]\jac(\bff^{\bA}, i), u_1 \ell_1
+ \dots + u_p \ell_p - 1 \right)\]
where:
\begin{itemize}
%    \item The original variables of the system $F$ are $\mathbf{X} \coloneqq
%        (X_1, \dots, X_n)$ and new variables $\mathbf{L} \coloneqq (L_1, \dots,
%        L_p)$ are introduced, for a total of $n+p$ variables,
    \item the entries of $\bsigma = \left( \sigma_1, \ldots, \sigma_{i-1}
        \right) $ and $\bu = (u_1, \ldots, u_p)$ are randomly drawn integer
        numbers of height bounded by $\widetilde{O}\left(\log(1/\epsilon) +
        n\log(d)\right)$ (see \cite[Section
        9.2]{elliottgiesbrechtschostgeneral}),
    \item $\bff^\bA$ indicates the sequence defined by $\bff({\bA}\bx)$ for
        $\bA$ being a matrix with integer coefficients of height bounded by
        $\widetilde{O}\left(\log(1/\epsilon) + n\log(d)\right)$ and \(\bx\)
        being the vector of variables  \(x_1, \ldots, x_{n}\). 

        In particular, letting \(b\) be the height of the original system
        $\bff$, the height of $\bff^{\bA}$ is in $\widetilde{O}\left(b +
        d\log(1/\epsilon) + dn\right)$ (see \cite[Section
        9.2]{elliottgiesbrechtschostgeneral}). 
%    \item $\jac(\bff^{\bA}, i)$ indicates the Jacobian matrix of $F^\mathbf{A}$
%        truncated from its first $i$ columns, that is the matrix of partial
%        derivatives of $F^\mathbf{A}$ for the variables $X_{i+1}, \dots, X_n$.
\end{itemize}
\noindent
\cite[Theorem 1]{safeyschostcomplexity} establishes bit complexity and output
size statements for \cite[Algorithm 2]{safeyschostcomplexity}. We provide below
a weakened version, dedicated to bi-affine polynomials (instead of the original
statement for multi-affine polynomials). 

\begin{theorem}\label{thm:compmultihomo}
\new{\emph{\cite[Theorem 1]{safeyschostcomplexity}}} Suppose that $F = (f_1,
\dots, f_N)$ is a sequence of multi-affine polynomials in $m$ groups of
variables $\mathbf{X}_1 = (X_{1}, \dots, X_{1,n})$ and  $\mathbf{Y} = (Y_{1},
\dots, Y_{m})$, such that $N = n + m$. Suppose that the height of each $f_i$ is
bounded by $s_i$ and the degree of each $f_i$ in $\mathbf{X}$ (resp.
\(\mathbf{Y}\)) is bounded by $d_{i}$ (resp. \(e_{i}\)). Suppose that $F$ is
given by means of a straight-line program $\Gamma$ of length $L$, that uses
integer coefficients of height at most $b$. \\ 
Then there exists an algorithm \emph{\cite[Algorithm
2]{safeyschostcomplexity}} taking $\Gamma, \bm{d} \coloneqq (d_{1},
e_{1}\dots, d_{N}, e_{N})$ and $\bm{s} \coloneqq (s_1, \dots, s_N)$ as
inputs and produces either a zero-dimensional rational parametrization of
the regular points of ${V}(F)$, with probability at least $21/32$, or a
parametrization of lesser degree, or \emph{FAIL}. In any case, the algorithm
performs 
    \[\widetilde{O}\left(Lb +
    \mathscr{C}_{\bm{n}}(\bm{d})\mathscr{H}_{\bm{n}}(\bm{\eta},
    \bm{d})(L+N\mathfrak{d}+N^2)N(\log(s)+N)\right)\] 
boolean operations, where
\[\mathfrak{d} \coloneqq \max\limits_{1 \leq i \leq N} d_{i} + e_{i}, 
    \ s \coloneqq \max\limits_{1 \leq i \leq N} s_i, \ \bm{\eta} =
\left(s_i + \log(n+1) d_i + \log(m + 1) e_i \right)_{1 \leq i \leq N},\]
$\mathscr{C}_{\bm{n}}(\bm{d})$ is the sum of the coefficients of the polynomial
\[\prod_{i=1}^N (d_{i}\theta_1 + e_{i}\theta_2) \mod
\left<\theta_1^{n+1}, \theta_2^{m+1}\right>,\] 
and $\mathscr{H}_{\bm{n}}(\bm{\eta}, \bm{d})$ is the sum of the coefficients of
the polynomial 
\[\prod_{i=1}^N (\eta_i \zeta + d_{i}\theta_1 + e_{i}\theta_2) \mod
\left<\zeta^2, \theta_1^{n+1}, \theta_2^{m+1}\right>.\] 
The output polynomials have degree at most $\mathscr{C}_{\bm{n}}(\bm{d})$ and
height $\widetilde{O}\left(\mathscr{H}_{\bm{n}}(\bm{\eta}, \bm{d}) +
N\mathscr{C}_{\bm{n}}(\bm{d})\right)$.
\end{theorem}

We apply these results to each system \(\polarsys{i}{\bff^\bA}, \bu\).
\new{Provided that $\bA$ and $\bsigma$ satisfy the genericity properties of 
\cite{elliottgiesbrechtschostgeneral}, these systems have finitely many
solutions, which are all regular, and hence \Cref{thm:compmultihomo} applies.}
Hence, we need to determine the quantities \(\mathscr{C}_{\bn}\left( \bd_i
\right) \) and $\mathscr{H}_{\bm{n}}(\bmeta_i, \bd_i)$ which are attached to
them according to their bi-affine structure that we described above. We obtain
the following estimates. 

\begin{lemma}\label{lem:Cnd}
    Reusing the notations introduced above, and for a matrix \(\bA\in
    \GL_n(\Zintegers)\) with entries of height in \(O\left(
    \log({1 / \epsilon}) + n \log(d) \right) \), 
    the following estimates hold for 
    the system $\polarsys{i}{\bff^\bA, \bu}$ (for $1 \leq i \leq
    n-p+1$):
    \begin{multline*}
        \mathscr{C}_{\bm{n}}(\bm{d}_i) = \binom{n-i}{p-1}d^p(d-1)^{n-p-i+1} \\
        \text{ and } 
    \mathscr{H}_{\bm{n}}(\bm{\eta}_i, \bm{d}_i) \in \widetilde{O}\left(n(b +
d\log(1/\epsilon) + dn)\binom{n}{p}d^{p}(d-1)^{n-p}\right). 
\end{multline*}
\end{lemma}
\begin{proof}
    We apply \Cref{thm:compmultihomo} to \(\polarsys{i}{\bff^\bA,
    \bu}\) using the bi-affine structure emphasized above in the introduction.
    Following its definition, $\mathscr{C}_{\bm{n}}(\bm{d}_i)$, is the sum of
    the coefficients of
    \begin{align*}
        %\!\begin{alignedat}[t]{2}
        &\phantom{={}} \prod_{j=1}^{n+p} (d_{j,1}\theta_1 + d_{j,2}\theta_2) 
        &&\mod \left<\theta_1^{n+1}, \theta_2^{p+1}\right> \\
        &= \prod_{j=1}^{i-1}(1\theta_1 + 0\theta_2)\prod_{j=1}^{p}(d\theta_1 + 
        0\theta_2)\prod_{j=1}^{n-i}((d-1)\theta_1 + 1\theta_2)\times(0\theta_1 +
        1\theta_2) &&\mod \left<\theta_1^{n+1}, \theta_2^{p+1}\right> \\
        &= d^p\theta_1^{p+i-1}\theta_2\prod_{j=1}^{n-i}((d-1)\theta_1 + 
        \theta_2) &&\mod \left<\theta_1^{n+1}, \theta_2^{p+1}\right> \\
        &= d^p\theta_1^{p+i-1}\theta_2 \binom{n-i}{p-1}(d-1)^{n-i-p+1}\theta_1
        ^{n-i-p+1}\theta_2^{p-1} &&\mod \left<\theta_1^{n+1}, \theta_2^{p+1}
        \right> \\
        &= \binom{n-i}{p-1}d^p(d-1)^{n-i-p+1}\theta_1^n\theta_2^p &&\mod 
        \left<\theta_1^{n+1}, \theta_2^{p+1}\right>.
        %\end{alignedat}
    \end{align*}
    We deduce that 
    \[\mathscr{C}_{\bm{n}}(\bm{d}_i) = \binom{n-i}{p-1}d^p(d-1)^{n-i-p+1}.\]
    Some routine computations show that for all \(2\le i\le n-p+1\), it holds
    that \(\mathscr{C}_{\bn}(\bd_1) \mathnew{\ge} \mathscr{C}_{\bn}(\bd_i)\).
    Following the definition of $\bm{\eta}$, we have for $i=1$:
    \begin{itemize}
        \item $\eta_1 = \cdots = \eta_p$ lie in 
            $\widetilde{O}\left(b + d\log(1/\epsilon) + dn + d\log(n+1)\right)
            \subset  
            \widetilde{O}\left(b + d\log(1/\epsilon) + dn\right)$,
        \item $\eta_{p+1} = \cdots = \eta_{n+p-1} $ lie in 
            $ \widetilde{O}\left(b +
            d\log(1/\epsilon) + dn + (d-1)\log(n+1) + \log(p+1)\right) $ which
            is in 
            $\widetilde{O}\left(b + d\log(1/\epsilon) + dn\right)$,
        \item $\eta_{n+p} \in 
            \widetilde{O}\left(\log(1/\epsilon) + n\log(d) + \log(p+1)\right) 
            \subset 
            \widetilde{O}\left(\log(1/\epsilon) + n\log(d)\right)$.
    \end{itemize}
    By definition, $\mathscr{H}_{\bm{n}}(\bm{\eta}_1, \bm{d}_1)$ is the sum of
    the coefficients of
    \begin{align*}
        &\prod_{j=1}^{n+p} (\eta_j \zeta + d_{j,1}\theta_1 + d_{j,2}\theta_2) 
        \mod \left<\zeta^2, \theta_1^{n+1}, \theta_2^{p+1}\right>\\
        = &\prod_{j=1}^{p}\left((b + d\log(1/\epsilon) + dn)\zeta + d\theta_1
        \right)\prod_{j=1}^{n-1}((b + d\log(1/\epsilon) + dn)\zeta + (d-1)
        \theta_1 + \theta_2) \\
        &\indent \times \left((\log(1/\epsilon) + n\log(d))\zeta + \theta_2
        \right) \mod \left<\zeta^2, \theta_1^{n+1}, \theta_2^{p+1}\right>.
    \end{align*}
    Expanding the polynomial yields the following sum of coefficients:
    \begin{align*}
        \mathscr{H}_{\bm{n}}(\bm{\eta}_1, &\bm{d}_1) = \binom{n-1}{p-1}d^p
        (d-1)^{n-p} \\
        &\indent + p(b + d\log(1/\epsilon) + dn)d^{p-1}(d-1)^{n-p}
        \left(\binom{n-1}{p-1} + \binom{n-1}{p-2}(d-1)\right) \\
        &\indent + \mathnew{(n-1)}(b + d\log(1/\epsilon) + dn)d^{p}(d-1)^{n-p-1}
        \left(\binom{n-2}{p-1} + \binom{n-2}{p-2}(d-1)\right) \\
        &\indent + (\log(1/\epsilon) + n\log(d))d^{p}(d-1)^{n-p-1}
        \left(\binom{n-1}{p} + \binom{n-1}{p-1}(d-1)\right) \\
        &\leq d^{p}(d-1)^{n-p}\biggl(p(b + d\log(1/\epsilon) + dn)
        \binom{n}{p-1} \\
        &\indent + \mathnew{n}(b + d\log(1/\epsilon) + dn)\binom{n-1}{p-1} + 
        (\log(1/\epsilon) + n\log(d))\binom{n}{p}\biggr) \\
        &\leq d^{p}(d-1)^{n-p}\mathnew{n}(b + d\log(1/\epsilon) + dn)
        (n+2)\binom{n}{p},
    \end{align*}
    where we have used the sum of consecutive binomials formula, and the fact
    that $\binom{n }{p-1} + \binom{n-1}{p-1} + \binom{n}{p} \leq
    (n+2)\binom{n}{p}$. We therefore obtain 
    \[\mathscr{H}_{\bm{n}}(\bm{\eta}_1, \bm{d}_1) \in
    \widetilde{O}\left(n\mathnew{^2}(b + d\log(1/\epsilon) + dn)\binom{n}{p}
    d^{p}(d-1)^{n-p}\right).\]
    \new{Moreover, incrementing $i$ changes a polynomial with associated values
    $\bd_i = (d-1,1)$ and $\eta_i \in \widetilde{O}\left(b + d\log(1/\epsilon) +
    dn\right)$ to one with values $(1,0)$ and $\widetilde{O}(n\log(d) +
    \log(1/\epsilon))$, which are all strictly smaller. Hence, we deduce that}
    \(\mathscr{H}_{\bm{n}}(\bm{\eta}_i, \bm{d}_i) \le  
    \mathscr{H}_{\bm{n}}(\bm{\eta}_1, \bm{d}_1)\) for \(2 \le i\le n-p+1\). 
\end{proof}

The next two statements are routine and already stated in a number of papers
(including \cite{safeyschostcomplexity, elliottgiesbrechtschostgeneral}). We
give their proofs for completeness.
\begin{lemma}\label{lem:SLPonepoly}
    Let $f \in \Zintegers[x_1, \dots, x_n]$ be a polynomial of degree $d$ of
    height \(b\). Then there exists a straight-line program of length $O\left(
    \binom{n+d}{d}\right)$, with all rational constants used therein of length
    at most \(b\), which evaluates $f$.
\end{lemma}
\begin{proof}
    We prove this by induction on $n$. For the base step, suppose that $f$ is a
    univariate polynomial of degree $d$. Then, by Horner's Method, we can
    evaluate $f$ from the value for $x_1$ and its coefficients in the monomial
    basis using $d$ multiplications and $d$ additions. In particular, there
    exists a straight-line program of length $2d \in O\left(\binom{d+1}{d}
    \right)$ evaluating $f$. 

    Suppose now that the statement holds in \(\ratfield[x_1, \ldots, x_{n-1}]\).
    Let $f\in \ratfield[x_1, \ldots, x_{n}]$ be of degree $d$. We can express
    $f$ as in \(\ratfield[x_1, \ldots, x_{n-1}][x_{n}]\):
    $$f = p_0x_n^d + p_1x_n^{d-1} + \cdots + p_{d-1}x_n + p_d,$$ where each
    $p_i\in \ratfield[x_1, \ldots, x_{n-1}]$ has degree at most $i$ or is \(0\).
    By the inductive hypothesis, they can each be computed by a straight-line
    program of length $O\left(\binom{n+i-1}{i}\right)$, and hence there exists a
    straight-line program of length $\sum_{i=0}^dO\left(\binom{n+i-1}{i}\right) = 
    O\left(\binom{n+d}{d}\right)$ computing all $p_i$'s. Once the $p_i$'s are
    obtained, it suffices to construct $f$ as if it was univariate in $x_n$,
    which requires an additional $2d$ steps. 

    Therefore, there exists a straight-line program of length
    $O\left(\binom{n+d}{d} + 2d\right)$ computing $f$, which is in
    $O\left(\binom{n+d}{d}\right)$ which completes the induction.
\end{proof}

\begin{lemma}\label{lem:SLP}
    Let $\bff = \left( f_1, \ldots, f_p \right)\subset \Zintegers[x_1, \ldots,
    x_n] $, \(\bA = \left( a_{i,j} \right) \in \GL_n(\Zintegers)\) and \(\bu =
    (u_1, \ldots, u_p) \subset \Zintegers^p\) with \(\deg(f_i) \le d\),
    \(\htt(f_i) \le  b\) and \(\htt(u_i) \le  b\) for \(1 \le i \le p\), and
    \(\htt(a_{i,j}) \le b\) for \(1 \le i,j \le n\).
    There exists a straight-line program $\Gamma$ of length $O\left(p\binom{n+d}
    {d} + \mathnew{n^2}\right)$ using integer constants of height in \(O(b)\)
    which evaluates $\polarsys{1}{\bff^\bA, \bu}$.
\end{lemma}
\begin{proof}
    Recall that
    \[\polarsys{1}{\bff^\bA, \bu} \coloneqq \left(\bff^\bA, [\ell_1 \cdots
    \ell_p]\jac(\bff^\bA, 1), u_1\ell_1 + \cdots + u_p\ell_p - 1\right).\]
    We apply \Cref{lem:SLPonepoly} to deduce the existence of a straight-line
    program which evaluates $\bff$, of length $O\left(p\binom{n+d}{d}\right)$,
    with integer constants of height \(\le b\). Applying the linear change of
    variables $\bA$ to $\bx$ beforehand requires at most an additional
    $O\left(n^2\right)$ steps. Therefore, there exists a straight-line program
    of length $L \in O\left(p\binom{n+d}{d} + n^2\right)$ which evaluates
    $\bff^\bA$. Again, its integer constants have height \(\le b\).

    By the Baur-Strassen theorem \cite[Theorem 1]{baurstrassen}, \new{applied
    component-wise}, this implies that there exists a straight-line program
    \(\Gamma'\) of length $\leq 3L$ computing $\left(\bff^\mathbf{A},
    \jac(\bff^\mathbf{A})\right)$. The construction of this new straight-line
    program (see e.g. the one of \cite[pages 78--81]{CKW11}) only adds a few
    gates to the program evaluating \(\bff^\bA\), all of them using the integer
    constants of that program. Hence, all integer constants in \(\Gamma'\) have
    height in  \(O(b)\). 

    Multiplying the Jacobian truncated of its first column by $[\ell_1 \cdots
    \ell_p]$ requires $p(n-1)$ multiplications and $(p-1)(n-1)$ additions, and
    computing the last polynomial requires $p$ multiplications and $p$
    additions. Since the entries of \(\bu\) have also a height  \(\le  b\), we
    deduce that this straight-line program uses integer constants lying in
    \(O(b)\). Combining these results, there exists a straight-line program
    $\Gamma_1$ computing the system $\polarsys{1}{\bff^\bA, \bu}$, with length
    \[3L+p(n-1)+(p-1)(n-1) + p + p \in O\left(p\binom{n+d}{d} + \mathnew{n^2}
    \right),\]
    as required.
\end{proof}

We can now prove \cref{thm:main,thm:main2} using these estimates. 

\begin{proof}[Proof of \cref{thm:main,thm:main2}]
    We start by estimating the cost of calling \cite[Algorithm
    2]{safeyschostcomplexity} with $\polarsys{1}{\bff^\bA, \bu}$ as input (for
    \(\bA\) and \(\bu\) chosen as in \cite[Algorithm
    1]{elliottgiesbrechtschostgeneral}) with heights we recalled above. We then
    apply \Cref{thm:compmultihomo} with the estimates of
    \cref{lem:Cnd} and \cref{lem:SLP}.
    
    Hence, we have:
    \begin{itemize}
        \item $\mathscr{C}_{\bm{n}}(\bm{d}) = \binom{n-1}{p-1}d^p(d-1)^{n-p}$, 
        \item $\mathscr{H}_{\bm{n}}(\bm{\eta}, \bm{d}) \in \widetilde{O}\left(
        n\mathnew{^2}(b + d\log(1/\epsilon) + dn)\binom{n}{p}d^{p}
        (d-1)^{n-p}\right)$, 
        \item $\mathfrak{d} = \max_{1 \leq j \leq n+p} d_{j,1} + d_{j,2} = d$,
        \item $s = \max_{1 \leq j \leq n+p} s_j \in \widetilde{O}\left(b + d
        \log(1/\epsilon) + dn\right)$, 
        \item $L \in O\left(p\binom{n+d}{d} + \mathnew{n^2}\right)$.
    \end{itemize}
    Moreover, by \Cref{lem:SLP} and the estimates from \cite[Section
    9.2]{elliottgiesbrechtschostgeneral}, the integer constants in the
    straight-line program evaluating \(\polarsys{1}{\bff^\bA, \bu}\) have
    heights in \(\widetilde{O}\left( b + d\log(1 / \epsilon) + dn \right) \). 
    Combining all this data with \Cref{thm:compmultihomo} yields
    a bit complexity of
    \begin{align*}
        \widetilde{O}&\left(\left(p\binom{n+d}{d} + \mathnew{n^2}\right)\left(b 
            + d\log(1/\epsilon) + dn\right) + n\mathnew{^2}(b + 
            d\log(1/\epsilon) + dn)\right. \\
        & \indent \binom{n-1}{p-1}\binom{n}{p}d^{2p}(d-1)^{2n-2p}\left(\left(
            p\binom{n+d}{d} + \mathnew{n^2}\right) + (n+p)d + (n+p)^2\right) \\
        & \indent (n+p)(\log\left(b + d\log(1/\epsilon) + dn\right) + n + p)
        \biggr) \\
        \subset  \widetilde{O}&\left(n\mathnew{^2}(b + d\log(1/\epsilon) + dn)
        \binom{n-1}{p-1}\binom{n}{p}d^{2p}(d-1)^{2n-2p}\left(\binom{n+d}{d} 
        + nd + n^2\right)n^2\right) \\
        \subset  \widetilde{O}&\left(n\mathnew{^4}(b + d\log(1/\epsilon) + dn)
        \binom{n}{p}^{\mathnew{2}}\binom{n+d}{d}d^{2p}(d-1)^{2n-2p}\right)
    \end{align*}
    for a single call to \cite[Algorithm 2]{safeyschostcomplexity} on
    $\polarsys{1}{\bff^\bA, \bu}$.
    Since \cite[Algorithm 1]{elliottgiesbrechtschostgeneral} performs
    $O\left(\log(n) + \log(1/\epsilon)\right)$ times such a call, since there
    are $n-p+1$ distinct $\polarsys{i}{\bff^\bA, \bu}$, and since, by
    \cref{lem:Cnd}, the quantities \(\mathscr{C}_{\bn}(\bd_i)\) and
    \(\mathscr{H}_{\bn}(\bmeta_i, \bd_i)\) are all dominated by the ones of
    \(\polarsys{1}{\bff^\bA, \bu}\), we conclude that the overall cost is 
    \begin{align*}
        \widetilde{O}&\left(n\mathnew{^4}(b + d\log(1/\epsilon) + dn)
        \binom{n-1}{p-1}\binom{n}{p}\binom{n+d}{d}d^{2p}(d-1)^{2n-2p}\right. \\
        & \indent \left(\log(n) + \log(1/\epsilon)\right)(n-p+1)\biggr) \\
        \subset \widetilde{O}&\left(n\mathnew{^4}(n-p+1)(b + d\log(1/\epsilon) 
        + dn)\log(1/\epsilon)\binom{n}{p}^{\mathnew{2}}\binom{n+d}{d}d^{2p}
        (d-1)^{2n-2p}\right).
    \end{align*} 
    Moreover, by \Cref{thm:compmultihomo}, the output polynomials
    have degree at most
    \[\mathscr{C}_{\bm{n}}(\bm{d}) = \binom{n-1}{p-1}d^p(d-1)^{n-p}\]
    and height in 
    \begin{align*}
        \widetilde{O}&\left(n\mathnew{^2}(b + d\log(1/\epsilon) + dn)
        \binom{n}{p}d^{p}(d-1)^{n-p} + (n+p)\binom{n-1}{p-1}d^p(d-1)^{n-p}
        \right) \\
        \subset  \widetilde{O}&\left(n\mathnew{^2}(b + d\log(1/\epsilon) + 
        dn)\binom{n}{p}d^{p}(d-1)^{n-p}\right). 
    \end{align*}

    \new{In the setting of \Cref{thm:main}, we assume that $d>2$, and it thus
    follows that the factors $n^4(n-p+1)$ and $n^2$ appearing in the above
    expressions are poly-logarithmic expressions of $d^{p}(d-1)^{n-p}$, and can
    hence be omitted in the $\widetilde{O}$ notation. Doing so yields the result
    of \Cref{thm:main}.}

    \new{In the setting of \Cref{thm:main2}, we assume that $d=2$, and hence
    $d-1 = 1$ and $\binom{n+d}{d} \in O(n^2)$. Performing these substitutions in
    the above expressions yields the result of \Cref{thm:main2}.}
\end{proof}

\paragraph*{Acknowledgements.} Mark Giesbrecht and \'Eric Schost acknowledge
support from the Natural Sciences and Engineering Research Council of Canada
(NSERC) Discovery Grant programme.  Edern Gillot and Mohab Safey El Din are
supported by ANR Project ANR-22-CE91-0007 “EAGLES” and the FA8655-25-1-7469 of
the European Office of Aerospace Research and Development of the Air Force
Office of Scientific Research (AFOSR).

%\begin{theorem}\label{thm:truecomp}
%    Let $F = (f_1, \dots, f_p)$ be a sequence of $p$ polynomials in $n$
%    variables, each of degree at most $d$ and of height at most $b$, and
%    suppose that $d \geq 2$ (the case $d=1$ being trivial). Then calling
%    \emph{\cite[Algorithm 2]{safeyschostcomplexity}} $O\left(\log(n) +
%    \log(1/\epsilon)\right)$ times on each input system $\mathcal{S}_1, \dots,
%    \mathcal{S}_{n-p+1}$ uses at most $$\widetilde{O}\left(n^3(b +
%    d\log(1/\epsilon) + dn)\log(1/\epsilon){n-1 \choose p-1}{n \choose p}{n+d
%    \choose d}d^{2p}(d-1)^{2n-2p}\right)$$ boolean operations. Moreover, the
%    polynomials in the outputs have degree at most \\
%    ${n-1 \choose p-1}d^p(d-1)^{n-p}$ and height $\widetilde{O}\left(n(b +
%d\log(1/\epsilon) + dn){n \choose p}d^{p}(d-1)^{n-p}\right)$.
%\end{theorem}

%In particular, both ${n \choose p}$ and ${n-1 \choose p-1}$ are polynomial in
%$p$ for a fixed value of $n$, and therefore the above complexity result, degree
%result and height result are all polynomial in $p$.

% \begingroup
% \phantomsection
% \addcontentsline{toc}{chapter}{References}
% \renewcommand{\bibname}{References}
% %\let\clearpage\relax
% \printbibliography
% \endgroup

% \bibliographystyle{plain}
% \bibliography{refs}

\begin{thebibliography}{99}
\bibitem{BGHM97}
B. Bank, M. Giusti, J. Heintz, and G. M. Mbakop,
\newblock Polar Varieties and Efficient Real Equation Solving: The Hypersurface Case,
\newblock {\em Journal of Complexity}, 13(1):5--27, 1997.
\newblock \url{https://doi.org/10.1006/jcom.1997.0432}

\bibitem{BGHM01}
B. Bank, M. Giusti, J. Heintz, and G. M. Mbakop,
\newblock Polar varieties and efficient real elimination,
\newblock {\em Mathematische Zeitschrift}, 238(1):115--144, 2001.
\newblock \url{https://doi.org/10.1007/PL00004896}

\bibitem{BGHP05}
B. Bank, M. Giusti, J. Heintz, and L. M. Pardo,
\newblock Generalized polar varieties: geometry and algorithms,
\newblock {\em Journal of Complexity}, 21(4):377--412, 2005.
\newblock \url{https://doi.org/10.1016/j.jco.2004.10.001}

\bibitem{BGHSS10}
B. Bank, M. Giusti, J. Heintz, M. Safey El Din, and \'E. Schost,
\newblock On the geometry of polar varieties,
\newblock {\em Applicable Algebra in Engineering, Communication and Computing}, 21:33--83, 2010.
\newblock \url{https://doi.org/10.1007/s00200-009-0117-1}

\bibitem{basupollackroy}
S. Basu, R. Pollack, and M.-F. Roy,
\newblock {\em Algorithms in Real Algebraic Geometry},
\newblock Springer, Heidelberg, 2006.
\newblock \url{https://doi.org/10.1007/3-540-33099-2}

\bibitem{baurstrassen}
W. Baur and V. Strassen,
\newblock The complexity of partial derivatives,
\newblock {\em Theoretical Computer Science}, 22(3):317--330, 1983.
\newblock \url{https://doi.org/10.1016/0304-3975(83)90110-X}

\bibitem{CKW11}
X. Chen, N. Kayal, and A. Wigderson,
\newblock Partial Derivatives in Arithmetic Complexity and Beyond,
\newblock {\em Foundations and Trends in Theoretical Computer Science}, 6(1--2):1--138, 2011.
\newblock \url{http://dx.doi.org/10.1561/0400000043}

\bibitem{collins}
G. E. Collins,
\newblock Quantifier elimination for real closed fields by cylindrical algebraic decomposition,
\newblock in {\em Automata Theory and Formal Languages}, Springer, 1975, pp. 134--183.
\newblock \url{https://doi.org/10.1007/3-540-07407-4_17}

\bibitem{elliottgiesbrechtschostgeneral}
J. Elliott, M. Giesbrecht, and \'E. Schost,
\newblock Bit complexity for computing one point in each connected component of a smooth real algebraic set,
\newblock {\em Journal of Symbolic Computation}, 116:72--97, 2023.
\newblock \url{https://doi.org/10.1016/j.jsc.2022.08.010}

\bibitem{GaSa24}
L. Gaillard and M. Safey El Din,
\newblock Solving parameter-dependent semi-algebraic systems,
\newblock in {\em Proc. ISSAC 2024}, ACM, 2024, pp. 447--456.
\newblock \url{https://doi.org/10.1145/3666000.3669718}

\bibitem{GreSa14}
A. Greuet and M. Safey El Din,
\newblock Probabilistic Algorithm for Polynomial Optimization over a Real Algebraic Set,
\newblock {\em SIAM Journal on Optimization}, 24(3):1313--1343, 2014.
\newblock \url{https://doi.org/10.1137/130931308}

\bibitem{Le_2021}
H. P. Le and M. Safey El Din,
\newblock Faster One Block Quantifier Elimination for Regular Polynomial Systems of Equations,
\newblock in {\em Proc. ISSAC 2021}, ACM, 2021.
\newblock \url{http://dx.doi.org/10.1145/3452143.3465546}

\bibitem{LeSa22}
H. P. Le and M. Safey El Din,
\newblock Solving parametric systems of polynomial equations over the reals through Hermite matrices,
\newblock {\em Journal of Symbolic Computation}, 112:25--61, 2022.
\newblock \url{https://doi.org/10.1016/j.jsc.2021.12.002}

\bibitem{SaSc17}
M. Safey El Din and \'E. Schost,
\newblock A Nearly Optimal Algorithm for Deciding Connectivity Queries in Smooth and Bounded Real Algebraic Sets,
\newblock {\em Journal of the ACM}, 63(6):48:1--48:37, 2017.
\newblock \url{https://doi.org/10.1145/2996450}

\bibitem{safeyschostcomplexity}
M. Safey El Din and \'E. Schost,
\newblock Bit complexity for multi-homogeneous polynomial system solving—Application to polynomial minimization,
\newblock {\em Journal of Symbolic Computation}, 87:176--206, 2018.
\newblock \url{https://doi.org/10.1016/j.jsc.2017.08.001}

\bibitem{safeyschostonepoint}
M. Safey El Din and \'E. Schost,
\newblock Polar varieties and computation of one point in each connected component of a smooth real algebraic set,
\newblock in {\em Proc. ISSAC 2003}, ACM, 2003, pp. 224--231.
\newblock \url{https://doi.org/10.1145/860854.860901}

\end{thebibliography}

% \input{elgigisascho.bbl}

\end{document}